\documentclass[submission]{eptcs}
\providecommand{\event}{E. Formenti (Ed.): AUTOMATA \& JAC 2012} 

\usepackage{graphicx}
\usepackage{amssymb,amsmath,stmaryrd,amsthm}
\usepackage[mathscr]{euscript}
\usepackage[latin1]{inputenc}
\usepackage{hyperref}
\usepackage{color}
\title{Intrinsic Simulations between Stochastic Cellular Automata}

\author{Pablo Arrighi\institute{Université de Grenoble (LIG, UMR 5217), France}
  \institute{Université de Lyon (LIP, UMR 5668), France}
	\and 	Nicolas Schabanel\institute{CNRS, Université Paris Diderot (LIAFA, UMR 7089), France}
        \institute{Université de Lyon (IXXI), France}
 	\and 	Guillaume Theyssier\institute{CNRS, Université de Savoie (LAMA, UMR 5127), France}
      }

\def\titlerunning{Intrinsic Simulations between Stochastic Cellular Automata}
\def\authorrunning{P.~Arrighi, N.~Schabanel and G.~Theyssier}

\newcommand\DELETE[1]{}

\newtheorem{definition}{Definition}
\newtheorem{lemma}{Lemma}
\newtheorem{corollary}{Corollary}
\newtheorem{theorem}{Theorem}
\newtheorem{fact}{Fact}

\newcommand{\PF}{\ensuremath{\operatorname{\mathscr{PF}}}}
\newcommand{\ZZ}{\ensuremath{\mathbb Z}}

\newcommand{\AUTO}[1]{{\ensuremath{\mathcal{#1}}}}

\newcommand\CAA{\AUTO A}
\newcommand\CAB{\AUTO B}

\newcommand\Mes[1]{\ensuremath{\mathcal{M}({#1})}}

\newcommand\Pof[1]{\ensuremath{\mathcal{P}({#1})}}
\newcommand\DET[1]{\ensuremath{\mathcal{D}_{#1}}}
\newcommand\NDET[1]{\ensuremath{\mathcal{N}_{#1}}}
\newcommand\STOC[1]{\ensuremath{\mathcal{S}_{#1}}}
\newcommand\cyl[2]{\ensuremath{[{#1}]_{#2}}}

\newcommand{\SUB}{\sqsubseteq}
\newcommand{\PROJ}{\unlhd}
\newcommand{\MIX}{{\PROJ{}\hskip-5pt\SUB}}
\newcommand{\NSUB}{\overset{N}{\SUB}}
\newcommand{\NPROJ}{\overset{N}{\PROJ}}
\newcommand{\NMIX}{\overset{N}{\MIX}}
\newcommand{\SSUB}{\overset{S}{\SUB}}
\newcommand{\SPROJ}{\overset{S}{\PROJ}}
\newcommand{\SMIX}{\overset{S}{\MIX}}
\newcommand{\DSUB}{\overset{D}{\SUB}}
\newcommand{\DPROJ}{\overset{D}{\PROJ}}
\newcommand{\DMIX}{\overset{D}{\MIX}}
\newcommand\rest[2]{{}_{#1}{#2}}
\newcommand\proj[2]{{}^{#1}{#2}}
\newcommand\mix[3]{{{}^{#1}_{#2}}{#3}}

\newcommand\shift[1]{\mathfrak{\sigma}_{#1}}
\newcommand\grp[2]{{#1}^{\langle#2\rangle}}
\newcommand\bloc[1]{b_{#1}}
\newcommand\debloc[1]{b^{-1}_{#1}}

\newcommand\simu{\preccurlyeq}
\newcommand\ssimui{{\simu_i^S}}
\newcommand\ssimus{{\simu_\pi^S}}
\newcommand\ssimum{{\simu_m^S}}
\newcommand\nsimui{{\simu_i^N}}
\newcommand\nsimus{{\simu_\pi^N}}
\newcommand\nsimum{{\simu_m^N}}
\newcommand\dsimui{{\simu_i^D}}
\newcommand\dsimus{{\simu_\pi^D}}
\newcommand\dsimum{{\simu_m^D}}

\newcommand\somerel\olessthan

\newcommand\PCA{\textsf{PlainPCA}}

 \renewcommand{\leq}{\leqslant}
 \renewcommand{\geq}{\geqslant}
 \renewcommand{\emptyset}{\varnothing}
 
 \newcommand{\anevent}{\operatorname{\mathcal E}}

\begin{document}
\maketitle

\begin{abstract}
The paper proposes a simple formalism for dealing with deterministic, non-deterministic and stochastic cellular automata in a unifying and composable manner. Armed with this formalism, we extend the notion of intrinsic simulation between deterministic cellular automata, to the non-deterministic and stochastic settings. We then provide explicit tools to prove or disprove the existence of such a simulation between two stochastic cellular automata, even though the intrinsic simulation relation is shown to be undecidable in dimension two and higher. The key result behind this is the caracterization of equality of stochastic global maps by the existence of a coupling between the random sources. We then prove that there is a universal non-deterministic cellular automaton, but no universal stochastic cellular automaton. Yet we provide stochastic cellular automata achieving optimal partial universality.
\end{abstract}


\section{Introduction} 

\paragraph{Motivations.} Cellular Automata (CA) are a key tool in simulating natural phenomena. This is because they constitute a privileged mathematical framework in which to cast the simulated phenomena, and they describe a massively parallel architecture in which to implement the simulator. 
Often however, the system that needs to be simulated is a noisy system. More embarrassingly even, it may happen that the system that is used as a simulator is again a noisy system. The latter is uncommon if one thinks of a classical computer as the simulator, but quite common for instance if one thinks of using a reduced model of a system as a simulator for that system.\\
Fortunately when both the simulated system and the simulating system are noisy, it could happen that both effects cancel out, i.e. that the noise of the simulator is made to coincide with that of the simulated. In such a situation a model of noise is used to simulate another, and the simulation may even turn out to be\ldots exact. This paper attempts to give a formal answer to the question: When can it be said that a noisy system is able to exactly simulate another?\\
This precise question has become crucial in the field of quantum simulation. Indeed, there are many quantum phenomena which we need to simulate, and these in general are quite noisy. Moreover, only quantum computers are able simulate them efficiently, but in the current state of experimental physics these are also quite noisy. Could it be that noisy quantum computers may serve to simulate a noisy quantum systems? The same remark applies to Natural Computing in general. Still, the question is challenging enough in the classical setting.

\paragraph{Challenges and results.} The first problem that one comes across is that stochastic CA have only received little attention from the theoretical community. When they have been considered, only probabilistic CA (PCA) consisting in a probabilistic function uniformly applied to a configuration have been studied \cite{Toom,Gacs,Fates,Mairesse,RST,FatesRST06}. However \cite{DBLP:conf/cie/ArrighiFNT11} exhibits several examples (such as the $\textrm{Parity}$ example which we use later) which cannot be realized as PCA, in spite of the fact that they require only local random correlation and hence fit naturally in the CA framework. Moreover \cite{DBLP:conf/cie/ArrighiFNT11} shows that the composition of two PCA is not always a PCA. The lack of composablity of a model is an obstacle for defining intrinsic simulation, because the notion must be defined up to grouping in space and in time. In \cite{DBLP:conf/cie/ArrighiFNT11} a composable model is suggested, but it lacks formalization.\\ 
In this paper we propose a simple formalism to deal with general stochastic CA. The formalism relies on considering a CA $F(c,s)$ fed, besides the current configuration $c$, with a new fresh independent uniform random configuration $s$ at every time step. This allows any kind of local correlations and includes in particular all the examples of \cite{DBLP:conf/cie/ArrighiFNT11}. As in turns out, the definition also captures deterministic and non-deterministic CA (non-deterministic CA are obtained by ignoring the probability distribution over the random configuration). More importantly, this formalism allows us to extend the notions of simulation developed for the deterministic setting \cite{bulking1,bulking2}, to the non-deterministic and stochastic settings. The choice of making explicit the random source in the formalism has turned out to be crucial to tackle the second problem, as it allows a precise analysis of the influence of randomness, in terms of simulation power. \\ 
Indeed the second problem that one comes across is that the question of whether two such stochastic CA are equal in terms of probability distributions is highly non-trivial. In particular, we show that testing if two stochastic CA define the same random map is undecidable in dimension $2$ and higher (Theorem~\ref{thm:deciding}). Still, we provide in section~\ref{sec:coupling} an explicit tool (the coupling of the random sources of two stochastic CA) that allows to prove (or disprove) the equality of their probability distributions. More precisely, we show that the existence of such a coupling is strictly equivalent to the equality of the distribution of the random maps of two stochastic CA (Theorem~\ref{thm:coupling}).\\
The choice of making expicit the random source allows us to show some no-go results. Any stochastic CA may only simulate stochastic CA with a compatible random source (where compatibility is expressed as a simple arithmetic equation, Theorem~\ref{thm:primefactors}). It follows that there is no universal stochastic CA (Corollary~\ref{cor:no:stoc:universal}). Still, we show that there is a universal CA for the non-deterministic dynamics (Theorem~\ref{thm:nondet:univ}), and we are able to provide a universal stochastic CA for every class of compatible random source (Theorem~\ref{thm:stoc:univ}).  

\paragraph{Plan.} Section~\ref{sec:basic} recalls the vital minimum about probability theory. Section~\ref{sec:SCA} states our formalism. Section~\ref{sec:coupling} gives tools to prove (or disprove) equality of stochastic global functions. Section~\ref{sec:simul} extends the notions of  intrinsic simulations to the non-deterministic and stochastic settings.  Section~\ref{sec:univ} provides the no-go results in the stochastic setting, the universality constructions. Section~\ref{sec:open} concludes this article with a list of open questions.

\section{Standard Definitions}
\label{sec:basic}

Even if this article focuses mainly on one-dimensional CA for the sake of simplicity, it extends naturally to higher dimensions.

For any finite set $A$ we consider the symbolic space $A^\ZZ$. For any $c\in A^\ZZ$ and $z\in\ZZ$ we denote by $c_z$ the value of $c$ at point $z$. $A^\ZZ$ is endowed with the Cantor topology (infinite product of the discrete topology on each copy of $A$) which is compact and metric (see \cite{kurkabook} for details). A basis of this topology is given by cylinders which are actually clopen sets: given some finite word $u$ and some position $z$, the cylinder $\cyl{u}{z}$ is the set
${\cyl{u}{z} = \{c\in A^\ZZ: \forall x, 0\leq x<|u|-1, c_{z+x}=u_x\}.}$

We denote by $\Mes{A^\ZZ}$ the set of Borel probability measures on $A^\ZZ$. By Carath\'eodory extension theorem, Borel probability measures are characterized by their value on cylinders. Concretely, a measure is given by a function $\mu$ from cylinders to the real interval $[0,1]$ such that $\mu(A^\ZZ) = 1$ and
\[\forall u\in Q^*, \forall z\in\ZZ, \quad \mu(\cyl{u}{z}) = \sum_{a\in A}\mu(\cyl{ua}{z}) = \sum_{a\in A}\mu(\cyl{au}{z-1})\]

We denote by $\nu_A$ the uniform measure over $A^\ZZ$ (s.t. $\nu_A(\cyl{u}{z})  = \frac{1}{|A|^{|u|}}$). We shall denote it as $\nu$ when the underlying alphabet~$A$ is clear from the context. 

We endow the set $\Mes{A^\ZZ}$ with the compact topology given by the following distance:
${\mathfrak{D}(\mu_1,\mu_2) = \sum_{n\geq 0}2^{-n}\cdot\max_{u\in A^{2n+1}} \bigl|\mu_1(\cyl{u}{-n})-\mu_2(\cyl{u}{-n})\bigr|}$.
See \cite{Pivato09} for a review of works on cellular automata from the measure-theoretic point of view.

\section{Stochastic Cellular Automata}
\label{sec:SCA}

Non-deterministic and stochastic cellular automata are captured by the same syntactical object given in the following definition. They differ only by the way we look at the associated global behavior. Moreover deterministic CA are a particular case of stochastic CA and can also be defined in the same formalism.


\subsection{The Syntactical Object}

\begin{definition}
A \emph{stochastic cellular automaton} $\AUTO A=(Q,R,V,V',f)$ consists in:
\begin{itemize}
\item a finite set of \emph{states} $Q$
\item a finite set $R$ called the \emph{random symbols}
\item two finite subsets of $\mathbb Z$: $V=\{v_1,\ldots, v_r\}$ and $V'=\{v'_1,\ldots, v'_{r'}\}$, called the \emph{neighborhoods}; $r$ and $r'$ are the sizes of the neighborhoods and ${\rho = \max_{v\in V\cup V'}|v|}$ is the \emph{radius} of the neighborhoods.
\item a \emph{local transition function} $f: Q^{r} \times R^{r'} \rightarrow Q$
\end{itemize}
A function $c\in Q^\ZZ$ is called a \emph{configuration}; $c_j$ is called the \emph{state} of the \emph{cell}~$j$ in configuration~$c$. A function $s\in R^\ZZ$ is called a \emph{$R$-configuration}.

In the particular case where $V'=\{0\}$ (i.e., where each cell uses its own random symbol only), we say that \CAA{} is a \emph{plain probabilistic cellular automaton} (\PCA{} for short).
\end{definition}

\begin{definition}[Explicit Global Function]
 To this local description, we associate the \emph{explicit global function} $F:Q^\ZZ\times R^\ZZ\rightarrow Q^\ZZ$ defined for any configuration $c$ and $R$-configuration $s$ by:
${F(c,s)_z = f\bigl((c_{z+v_1},\ldots,c_{z+v_r}),(s_{z+v'_1},\ldots,s_{z+v'_{r'}})\bigr).}$
Given a sequence $\bigl(s^t\bigr)_t$ of $R$-configurations and an initial configuration $c$, we define the associated \emph{space-time diagram} as the bi-infinite matrix ${\bigl(c^t_z\bigr)_{t\geq 0, z\in\ZZ}}$ where $c^t\in Q^\ZZ$ is defined by $c^0=c$ and ${c^{t+1}=F(c^t,s^t)}$. We also define for any $t\geq 1$ the $t^\text{th}$ iterate of the explicit global function $F^t:Q^\ZZ\times \bigl(R^\ZZ\bigr)^t\rightarrow Q^\ZZ$ by $F^0(c)=c$ for all configuration~$c$ and 
\[F^{t+1}(c,s^1,\ldots,s^{t+1}) = F\bigl(F^t(c,s^1,\ldots,s^t),s^{t+1}\bigr)\]
so that ${c^t=F^t(c,s^1,\ldots,s^t)}$.\DELETE{ We denote by $\CAA^t = (Q,R^t,V_+,V'_+,f^t)$ the stochatic CA whose explicit global function is $F^t$.}
\end{definition}

In this paper, we adopt the convention that local functions are denoted by a lowercase letter (typically $f$) and explicit global functions by the corresponding capital letter (typically $F$). Moreover, we will often define CA through their explicit global function since details about neighborhoods often do not matter in this paper.

The explicit global function capture all possible actions of the automaton on configurations. This function allows to derive three kinds of dynamics: deterministic, non-deterministic and stochastic. 




\subsection{Deterministic and Non-Deterministic Dynamics}

\paragraph{Deterministic.}  The \emph{deterministic global function} ${\DET{F}:Q^\ZZ\rightarrow Q^\ZZ}$ of ${\CAA=(Q,R,V,V',f)}$ is  defined by ${\DET{F}(c) = F(c,0^\ZZ)}$ where $0$ is a distinguished element of $R$. \CAA{}  is said to be \emph{deterministic} if its local transition function~$f$ does not depend on its second argument (the random symbols).

\paragraph{Non-Deterministic.} The \emph{non-deterministic global function} ${\NDET{F}:Q^\ZZ\rightarrow \Pof{Q^\ZZ}}$ of~$\CAA$ is  defined for any configuration $c\in Q^\ZZ$ by 
${\NDET{F}(c) = \{F(c,s):s\in R^\ZZ\}}$.

\paragraph{Dynamics.} The deterministic dynamics of \CAA{} is given by the sequence of iterates $(\DET{F}^t)_{t\geq 0}$. Similarly the non-deterministic dynamics of $\CAA$ is given by the iterates $\NDET{F}^t:Q^\ZZ\rightarrow \Pof{Q^\ZZ}$ defined by $\NDET{F}^{0}(c) = \{c\}$ and $\NDET{F}^{t+1}(c) = \bigcup_{c'\in\NDET{F}^t(c)}\NDET{F}(c')$. 


\subsection{Stochastic Dynamics}

 The stochastic point of view consists in taking the $R$-component as a source of randomness. More precisely, the explicit global function $F$ is fed at each time step with a
random uniform and independent $R$-configuration. This defines a stochastic process for which we are then interested in the distribution of states across space and time. By Carath\'eodory extension theorem, this distribution is fully determined by the probabilities of the events of the form  ``starting from $c$, the word $u$ occurs at position $z$ after $t$ steps of the process''.
Formally, for $t=1$, this event is the set:
\[\anevent_{c,\cyl{u}{z}}=\bigl\{s\in R^\ZZ : F(c,s)\in\cyl{u}{z}\bigr\}.\]
In order to evaluate the probability of this event, we use the locality of the explicit global function $F$. The event ``$F(c,s)\in\cyl{u}{z}$'' only depends of the cells of $s$ from position ${a=z-\rho}$ to position ${b=z+\rho+|u|-1}$. Therefore, if ${J = \{v\in R^{b-a}:F(c,[v]_a)\subseteq[u]_z\}}$, then $\anevent_{c,\cyl{u}{z}}=\cup_{v\in J} [v]_a$ and hence $\anevent_{c,\cyl{u}{z}}$ is a measurable set of probability: $\nu_R(\anevent_{c,\cyl{u}{z}})=\sum_{v\in J}\nu_R(\cyl{v}{a})=|J|/|R|^{b-a}$ (recall that $\nu_R$ is the uniform measure over $R^\ZZ$).


More generally to  any CA \CAA{} we associate its \emph{stochastic global function} ${\STOC{F}:Q^\ZZ\rightarrow\Mes{Q^\ZZ}}$ defined for any configuration $c\in Q^*$ by: $\forall u\in Q^\ZZ, \forall z\in\ZZ,$
 \[
\bigl(\STOC{F}(c)\bigr)(\cyl{u}{z}) = \nu_R(\anevent_{c,\cyl{u}{z}}) \text{ = the probability of event $\anevent_{c,\cyl{u}{z}}$.}
\]

\paragraph{Example.} For instance, consider the stochastic function $\textrm{Parity}$ that maps every configuration~$c$ over the alphabet $\{0,1,\#\}$ to a random configuration in which every $\{0,1\}$-word of length~$\ell$ delimited by two consecutive $\#$ in~$c$ is replaced by a random independent uniform word of length~$\ell$ with even parity. This cannot be realized by a \PCA{}. Still, one can realize the stochastic function $\textrm{Parity}$ as a stochastic CA by means of a local transition function~$f$ of the above type, as follows. Given the configuration~$c$ and a uniform random $\{0,1\}$-configuration~$s$, $f(c_{-1},c_0,c_1, s_{-1},s_0)$ is: $\#$ if $c_0=\#$; else $s_0$ if $c_{-1} = \#$; else $s_{-1}$ if $c_1=\#$; and $(s_{-1}+s_0 \mod 2)$ otherwise. One can easily check that this local correlation ensures that every word delimited by two consecutive $\#$ is indeed mapped to a uniform independent random word of even parity.

\paragraph{Dynamics.} As opposed to the deterministic and non-deterministic setting, defining an iterate of this map is a not so trivial task.  There are two approaches: defining directly the measure after $t$ steps or extending the map $\STOC{F}$ to a map from $\Mes{Q^\ZZ}$ to itself. Both rely crucially on the continuity of $F$. In particular, we want to make sure that the definition of the measure after $t$ steps matches $t$ iterations of the one-step map, and hence, is independent of the explicit mechanics of $F$ but depends only on the map $\STOC{F}$ defined by $F$. 

The easiest one to present is the first approach. For any $t\geq 1$, the event $\anevent^t_{c,\cyl{u}{z}}$ that the word $u$ appears at position~$z$ at time~$t$ from configuration~$c$ consists in the set of all $t$-uples of random configurations $(s^1,\ldots,s^t)$ yieldings $u$ at position~$z$ from~$c$, i.e.: 
\[\anevent^t_{c,\cyl{u}{z}}=\bigl\{(s^1,\ldots,s^t)\in\bigl(R^\ZZ\bigr)^t : F^t(c,s^1,\ldots,s^t)\in\cyl{u}{z}\bigr\}\]
As before $\anevent^t_{c,\cyl{u}{z}}$ is a measurable set in $\bigl(R^\ZZ\bigr)^t$ because it is a product of finite unions of cylinders by the locality of $F$. We therefore define $\STOC{F}^t:Q^\ZZ\rightarrow\Mes{Q^\ZZ} $, the iterate of the stochastic global function, by:
\[\bigl(\STOC{F}^t(c)\bigr)(\cyl{u}{z}) = \nu_{R^t}(\anevent^t_{c,\cyl{u}{z}})
\text{ =  the probability of event $\anevent^t_{c,\cyl{u}{z}}$}
\]
where $\nu_{R^t}$ denotes the uniform measure on the product space $\bigl(R^\ZZ\bigr)^t$. For similar reasons as above, $\STOC{F}^t(c)$ is a well-defined probability measure.

The following key technical fact ensures that two automata define the same distribution over time as soon as their one-step distributions match.
\begin{fact}\label{fac:iterates}
  Let $\AUTO{A}$ and $\AUTO{B}$ be two stochastic CA with the same set of states $Q$ (and possibility different random alphabet) and of explicit global functions $F$ and $G$ respectively. If $\STOC{F}=\STOC{G}$ then for all $t\geq 1$ we have $\STOC{F}^t=\STOC{G}^t$
\end{fact}

\begin{proof}
  Consider a CA of stochastic global function $F$. Consider a word $u$ and a position $z$. Let $\phi:Q^\ZZ\rightarrow \Pof{R^\ZZ}$ be function that associates to a configuration~$c$ the event $\anevent_{c,\cyl uz}$. $\phi(c)$ is entirely determined by the states of the cells from positions $a=z-\rho$ to $b=z+|u|+\rho$ in~$c$ (locality of $F$). Therefore $\phi$ is constant over every cylinder $\cyl{v}{a}$ with $v\in Q^{b-a}$. If we distinguish some $c_v\in\cyl va$ for every $v\in Q^{b-a}$, we obtain by definition of $F^t$ and  continuity of $F$:
\[
 \anevent^{t+1}_{c,\cyl{u}{z}} = \bigcup_{v\in Q^{b-a}} \left(\anevent^t_{c,\cyl{v}{a}}\times \anevent_{c_v,\cyl{u}{z}}\right).
 \]
Then, since sets $\bigl(\anevent^t_{c,\cyl{v}{a}}\bigr)_{v\in Q^{b-a}}$ are pairwise disjoint (because $F$ is deterministic and cylinders $\cyl{v}{a}$ are pairwise disjoint), we have
\begin{align*}
  \bigl(\STOC{F}^{t+1}(c)\bigr)(\cyl{u}{z}) =
  \nu_{R^{t+1}}(\anevent^{t+1}_{c,\cyl{u}{z}}) &=
  \sum_{v\in Q^{b-a}} \nu_{R^{t+1}}\bigl(\anevent^t_{c,\cyl{v}{a}}\times
  \anevent_{c_v,\cyl{u}{z}}\bigr)\\ &=
  \sum_{v\in Q^{b-a}}\nu_{R^t}(\anevent^t_{c,\cyl{v}{a}})\cdot\nu_R(\anevent_{c_v,\cyl{u}{z}})\\
  &= \sum_{v\in Q^{b-a}} \bigl(\STOC{F}^{t}(c)\bigr)(\cyl{v}{a})\cdot\bigl(\STOC{F}(c_v)\bigr)(\cyl{u}{z})
\end{align*}
The value of $S_F^t(c)$ over cylinders  can thus be expressed recursively as a function of a finite number of values $S_F$ over a finite number of cylinders. It follows that if for some pair of CA \CAA\/ and \CAB\/ with explicit global functions $F$ and $G$ we have  $\STOC{F}=\STOC{G}$, then $\STOC{F}^{t}=\STOC{G}^{t}$ for all $t$.
\end{proof}

In our setting one can recover the non-deterministic dynamics from the stochastic dynamics of a given stochastic CA. This heavily relies on the continuity of explicit global functions and compacity of symbolic spaces.

\begin{fact}\label{fac:stocndet}
  Given two CA with same set of states and explicit global functions $F_A$ and $F_B$, if $\STOC{F_A}=\STOC{F_B}$ then $\NDET{F_A}=\NDET{F_B}$.
\end{fact}
\begin{proof}
  Given some stochastic CA of explicit global function $F$, some configuration $c$ and some cylinder $\cyl{u}{z}$ we have 
\[\NDET{F}(c)\cap\cyl{u}{z}\neq\emptyset\Leftrightarrow \anevent_{c,\cyl{u}{z}}\neq\emptyset\Leftrightarrow\bigl(\STOC{F}(c)\bigr)(\cyl{u}{z})>0\]
by definition of $\NDET{F}$ and $\STOC{F}$. Since $\NDET{F}(c)$ is a closed set (continuity of $F$) it is determined by the set of cylinders intersecting it (compacity of the space). Hence $\NDET{F}(c)$ is determined by $\STOC{F}(c)$. The lemma follows. 
\end{proof}


\section{Equality of random maps: undecidability and explicit tools}
\label{sec:coupling}

\paragraph{An undecidable task for dimension~$2$ and higher.}
In the classical deterministic case, it is easy to determine whether two CA have the same global function. Equivalently determining whether two stochastic CA, as syntactical objects, have the same explicit global functions $F$ and $G$ is easy. However, given two stochastic CA which have possibly different explicit global function $F$ and $G$, it still happen that $\NDET{F}=\NDET{G}$ or $\STOC{F}=\STOC{G}$, and determining whether this is the case turns out to be a difficult problem. In fact, Theorem~\ref{thm:deciding} states that these two decision problems are at least as difficult as the surjectivity problem of classical CA, which is undecidable in dimension 2 and higher \cite{Kari94}. 

\begin{theorem}\label{thm:deciding}
  Let $\mathcal{P}_N$ (resp. $\mathcal{P}_S$) be the problem of deciding whether two given stochastic CA have the same non-deterministic (resp. stochastic) global function. The surjectivity problem of classical deterministic CA is reducible to both $\mathcal{P}_N$ and $\mathcal{P}_S$.
\end{theorem} 
\newcommand{\proofthdeciding}{
\begin{proof}
  Consider a classical CA $F:Q^{\ZZ}\rightarrow Q^{\ZZ}$ and define $\mu_F$ as the image by $F$ of the uniform measure $\mu_0$ on $Q^{\ZZ}$:
  \[\mu_F(\cyl{u}{z}) = \nu\bigl(F^{-1}(\cyl{u}{z})\bigr)\]
  It is well-known that $F$ is surjective if and only if ${\mu_F=\nu}$ (this result is true in any dimension, the proof for dimension 1 is in \cite{Hedlund} and follows from \cite{maki} for higher dimensions, but we recommend \cite{Pivato09} for a modern exposition in any dimension).
  
  Now let us define the stochastic CA $\CAA = (Q,Q,V,V',g)$ such that ${G(c,s) = F(s)}$. With this definition, \CAA{} is such that, for all $c$, ${\STOC{G}(c) = \mu_F}$. Hence, ${\STOC{G}(c)}$ is the uniform measure for any $c$ if and only if $F$ is surjective. We have also ${\NDET{G}(c)=Q^\ZZ}$ for all $c$ if and only if $G$ is surjective. The theorem follows since \CAA{} is recursively defined from $F$. 
\end{proof}
}\proofthdeciding

\paragraph{Explicit tools for (dis)proving equality.}
Even if testing the equality of the non-deterministic or stochastic dynamics of two stochastic CA is undecidable for dimension $2$ and higher, Theorem~\ref{thm:coupling} states that equality, when it holds, can always be certified in terms of a \emph{stochastic coupling}. Indeed the stochastic coupling, by matching their two source of randomness, serves as a witness of the equality of the stochastic CA. This provides us with a very useful technique, because the existence of such a coupling is easy to prove or disprove in many concrete examples. Again the result heavily relies on the continuity of the explicit global function~$F$.

Let us first recall the standard notion of coupling.
\begin{definition}
  Let $\mu_1\in\Mes{Q_1^\ZZ}$ and $\mu_2\in\Mes{Q_2^\ZZ}$. A coupling of $\mu_1$ and $\mu_2$ is a measure $\gamma\in \Mes{Q_1^\ZZ\times Q_2^\ZZ}$ such that for any measurable sets $\anevent_1$ and $\anevent_2$, ${\gamma(\anevent_1\times Q_2^\ZZ)=\mu_1(\anevent_1)}$ and ${\gamma( Q_1^\ZZ\times\anevent_2)=\mu_2(\anevent_2)}$.
\end{definition}
Concretely, a coupling couples two measures so that each is recovered when the other is ignored. The motivation in defining a coupling is to bound the two distributions in order to prove that they induce the same kind of behavior: for instance, one can easily couple the two uniform measures over $\{1,2\}$ and $\{1,2,3,4\}$ so that with probability~$1$, both numbers will have the same parity ($\gamma$ gives a probability $1/4$ to each pair $(1,1)$, $(2,2)$, $(1,3)$ and $(2,4)$ and $0$ to the others). This demonstrates that the parity function is identically distributed in both cases.  

Theorem~\ref{thm:coupling} states that the dynamics of two stochastic CA are identical if and only if their is a  coupling of their random configurations so that their stochastic global functions become almost surely identical. This is one of our main results. 
\begin{definition}
  Two stochastic cellular automata, $\CAA_1=(Q,R_1,V_1,V_1',f_1)$ and $\CAA_2=(Q,R_2,V_2,V_2',f_2)$, with the same set of states~$Q$ are \emph{coupled} on configuration $c\in Q^\ZZ$ by a measure $\gamma\in\Mes{R_1^\ZZ\times R_2^\ZZ}$ if 
  \begin{enumerate}
  \item $\gamma$ is a coupling of the uniform measures on $R_1^\ZZ$ and $R_2^\ZZ$;
  \item $\gamma\bigl(\{(s_1,s_2)\in R_1^\ZZ\times R_2^\ZZ : F_1(c,s_1)=F_2(c,s_2)\}\bigr) = 1$, i.e. $F_1$ and $F_2$ produce almost surely the same image when fed with the $\gamma$-coupled random sources.
  \end{enumerate}
  
\end{definition}
Note that the set of pairs $(s_1,s_2)$ defined above is measurable because it is closed ($F_1$ and $F_2$ are continuous).

\begin{theorem}\label{thm:coupling}
  Two stochastic CA with the same set of states have the same stochastic global function if and only if, on each configuration $c$, they are coupled by some measure $\gamma_c$ (which depends on~$c$).
\end{theorem}
\noindent {\em Outline of the proof.} We fix a configuration $c$. By continuity of the explicit global functions, we construct a sequence of partial couplings $(\gamma_c^n)$ matching the random configurations of finite support of radius $n$. We then extract the coupling $\gamma_c$ from $(\gamma_c^n)$ by compacity of $\Mes{R_A^\ZZ\times R_B^\ZZ}$.

{
\begin{proof}
{\em Details.} First, if $\CAA_1=(Q,R_1,V_1,V_1',f_1)$ and $\CAA_2=(Q,R_2,V_2,V_2',f_2)$ are coupled by $\gamma_c$ on configuration $c\in Q^\ZZ$, consider for any cylinder $\cyl{u}{z}$ the sets
  \begin{align*}
    \anevent_1 &= \{s\in R_1^\ZZ : F_1(c,s)\in\cyl{u}{z}\}\\
    \anevent_2 &= \{s\in R_2^\ZZ : F_2(c,s)\in\cyl{u}{z}\}\\
    X &= \{(s_1,s_2) \in R_1^\ZZ\times R_2^\ZZ : F_1(c,s_1)=F_2(c,s_2)\}
  \end{align*}
  Then, by the property of the coupling by $\gamma_c$,  we have 
\[\bigl(\STOC{F_1}(c)\bigr)(\cyl{u}{z}) = \nu_1(\anevent_1) = \gamma_c(\anevent_1\times R_2^\ZZ)\]
 where $\nu_1$ is the uniform measure on $R_1^\ZZ$. But $\gamma_c(\anevent_1\times R_2^\ZZ) = \gamma_c\bigl((\anevent_1\times R_2^\ZZ) \cap X\bigr)$ since $\gamma_c(X)=1$. Symmetrically we have \[\bigl(\STOC{F_2}(c)\bigr)(\cyl{u}{z}) = \gamma_c\bigl((R_1^\ZZ\times \anevent_2) \cap X\bigr).\]
 But, by definition of sets $\anevent_1$, $\anevent_2$ and $X$, we have ${R_1^\ZZ\times \anevent_2 \cap X = \anevent_1\times R_2^\ZZ \cap X}$. We conclude that $\STOC{F_1}=\STOC{F_2}$.
 For the other direction of the theorem, suppose $\STOC{F_1}=\STOC{F_2}$ and fix some configuration $c$. We denote by $\mu$ the measure ${\STOC{F_1}(c)=\STOC{F_2}(c)}$. Without loss of generality we can suppose that $\CAA_1$ and $\CAA_2$ have same raidii $\rho$: $\rho_1=\rho_2=\rho$. We construct a sequence $(\gamma^n)$ of measures from which we can extract a limit point (by compacity of the space of measures) which is a valid coupling of $\CAA_1$ and $\CAA_2$ on configuration $c$.
\newcommand\ccyl[1]{[{#1}]}
To simplify the proof we focus on centered cylinders: for any word $w$ of odd length, we denote by $\ccyl{w}=\cyl{w}{z_w}$ where ${z_w = -\frac{|w|-1}{2}}$.
Let's fix $n$. For any word $u\in Q^{2n+1}$ we define:
\begin{align*}
  S_u^1 &= \{s\in R_1^\ZZ : F_1(c,s)\in\ccyl{u}\}\\
  S_u^2 &= \{s\in R_2^\ZZ : F_2(c,s)\in\ccyl{u}\}
\end{align*}
$F_1$ and $F_2$ being of radius $\rho$ we can write $S_u^i$ as a finite union of centered cylinders of length ${2(n+\rho)+1}$:
\[S_u^i = \bigcup_{v\in P_u^i}\ccyl{v}\]
where ${P_u^i\subseteq R_i^{2(n+\rho)+1}}$. Define the following partition $I_u^i$ of the real interval $[0,1)$ by:
\newcommand\ord[1]{\mathsf{rank}(#1)}
\[I_u^i(v) = \left[\frac{\ord{v}}{\# P_u^i};\frac{\ord{v}+1}{\# P_u^i}\right[\]
where $\ord{v}\in\{0,\ldots,\# P_u ^i - 1\}$ is the rank of $v$ in some arbitrarily chosen total ordering of $P_u^i$ (the lexicographical order for instance).
Since the sets $P_u^i$ form a partition of ${R_i^{2(n+\rho)+1}}$ when $u$ ranges over all words of $Q^{2n+1}$, we have for any $v\in R_i^{2(n+\rho)+1}$:
\[|I_u^i(v)| = \frac{1}{\# P_u^i} = \frac{\nu_i(\ccyl{v})}{\mu(\ccyl{u})}\]
(recall that $\nu_i$ stands for the uniform measure over $R_i^\ZZ$).
Now, for every ${v^1\in R_1^{2(n+\rho)+1}}$ and ${v^2\in R_2^{2(n+\rho)+1}}$, we construct $\gamma^n$ as:
\[\gamma^n(\ccyl{v^1},\ccyl{v^2}) =
\begin{cases}
  |I_u^1(v^1)\cap I_u^2(v^2)|\cdot\mu_Q(\ccyl{u}) &\text{ if $\exists u$ s.t. ${v^i\in P_u^i}$ for both $i=1,2$}\\
  0 &\text{ otherwise.}
\end{cases}\]
Furthermore, if $0_i$ is a distinguished element of $R_i$, we extend the definition of $\gamma^n$ to any pair $v^1\in R_1^{2m+1}$ and $v^2\in R_2^{2m+1}$ with ${m\geq n+\rho}$ by:
\[\gamma^n(\ccyl{v^1},\ccyl{v^2}) =
\begin{cases}
  \gamma^n(\ccyl{w^1},\ccyl{w^2}) &\text{ if $v^i = 0_i^{m-n-\rho}w^i0_i^{m-n-\rho}$ for $i=1,2$}\\
  0 &\text{ else.}
\end{cases}
\]
By $\sigma$-additivity $\gamma^n$ is thus defined on any cylinder and by extension theorem is a well-defined measure. Now by construction, we have for any $v^1\in R^{2(n+\rho)+1}$:
\[\gamma^n(\ccyl{v^1},R_2^\ZZ) = |I_u^1(v^1)|\cdot\mu_Q(\ccyl{u})\]
for some $u$ such that ${v^1\in P_u^1}$. Hence, ${\gamma^n(\ccyl{v^1},R_2^\ZZ) =\nu_1(\ccyl{v^1})}$. By $\sigma$-additivity of $\gamma^n$ and $\mu$, this equality holds for any $v^1\in R^{2m+1}$ with $m\leq n+\rho$. Symmetrically we have ${\gamma^n(R_1^\ZZ,\ccyl{v^2}) =\mu_{U_2}(\ccyl{v^2})}$ for any $v^1\in R^{2m+1}$.
Moreover, by definition, $\gamma^n(\ccyl{v^1},\ccyl{v^2}) = 0$ if there is no $u$ such that ${v^i\in P_u^i}$ for $i=1,2$. We deduce that the set:
\[X_n = \bigcup_{u\in Q^{2n+1}} S_u^1\times S_u^2\]
has measure $1$. More precisely, since $X_{n+1}\subseteq X_n$, for any ${m\leq n}$, ${\gamma^n(X_m)=1}$.\\
To conclude the proof, let $\gamma$ be any limit point of the sequence $(\gamma^n)_n$. By the definition of the distance on the space of measures, we have:
\begin{enumerate}
\item $\forall m,\forall w\in R_1^{2m+1}, \gamma(\ccyl{w},R_2^\ZZ) = \nu_1(\ccyl{w})$ and symmetrically for the $R_2$ component, hence $\gamma$ is a coupling of uniform measure on $R_1^\ZZ$ and $R_2^\ZZ$;
\item $\forall n,\gamma(X_n)=1$ hence $\gamma(\cap_n X_n) = 1$ where
\[\bigcap_n X_n = X = \{(s_1,s_2) \in R_1^\ZZ\times R_2^\ZZ : F_1(c,s_1)=F_2(c,s_2)\}\]
\end{enumerate}
We deduce that $\CAA_1$ and $\CAA_2$ are coupled on $c$ by measure $\gamma$.
 
\end{proof}
}

Notice that the proof of this theorem is non-constructive (recall that equality of stochastic global maps is undecidable in dimension~$2$ and higher). Moreover, it is easy to get convinced on a simple example that the coupling must depend on the configuration. Consider the two following automata with states $Q=R=\{0,1\}$ and neighborhoods $V=V'=\{0\}$: \CAA~with explicit global function $F(c,s)=s$ and \CAB\/~with explicit global function $G(c',s')=c'+s'\mod 2$. Clearly, both \CAA\/ and \CAB\/ define the same blank noise CA and the coupling proving this fact is defined for all $z\in\ZZ$ and all $a,b\in\{0,1\}$ by $\gamma_c(\cyl az,\cyl bz) = 1/2$ if and only if $a = b+c_z \mod 2$, and $=0$ otherwise. This coupling demonstrates indeed that $\gamma_c\bigl(\{(s,s'):F(c,s) = G(c,s')\}\bigr)=1$ yielding that the dynamics are identical; but note that $\gamma_c$ must depend on~$c$.

\DELETE{\begin{proof}(Sketch)
Given a configuration $c$, the principle is for every $n$ and every word $u$ of length $2n+1$ to build an sequence of increasing finite support couplings $(\gamma_c^n)$ where $\gamma_c^n$ matches the two finite subsets of finite random strings (thanks to the locality of \CAA\/ and \CAB)  that yields the word $u$ centered on $0$ in $\CAA$ and $\CAB$. We then conclude by extracting  from $(\gamma_c^n)$ a subsequence converging to some measure $\gamma_c$ thanks to the compacity of the metric space $\Mes{R_A^\ZZ\times R_B^\ZZ}$. As pointed out earlier, this proof is non-constructive.
\end{proof}}


\section{Intrinsic Simulations}
\label{sec:simul}

The purpose of this section is to give a precise meaning to the sentence ``$\AUTO{A}$ is able to simulate $\AUTO{B}$'' or equivalently ``$\AUTO{A}$ contains  the behavior of $\AUTO{B}$''. 

Our approach follows a series of works on simulations between classical deterministic CA \cite{phdrapaport,phdollinger,phdtheyssier,bulking1,bulking2}.
We are going to define simulation pre-orders on stochastic CA which extend the simulation pre-orders defined over classical deterministic CA in \cite{bulking2}. Precisely, we want the new pre-order to be exactly the classical pre-order when restricted to deterministic CA. For general background and motivation behind this simulation pre-order approach we refer to \cite{bulking1,bulking2}. Intrinsic simulation has also been brought to deterministic quantum CA in \cite{ArrighiNUQCA}.

In each case (the deterministic, the non-deterministic, and the stochastic global functions), we will define simulation as an equality of dynamics up to some \emph{local} transformations.

\subsection{Transformations}

The transformations we consider are natural stochastic extensions of the transformation defined in \cite{bulking1,bulking2} for the  classical deterministic CA. These transformations can be divided into two categories: \emph{trimming operations} which allow to trim unwanted parts off the dynamics, and \emph{rescaling transformations} which augment the set of states and/or the neighborhoods.

\subsubsection{Trimming operations} 

They are based on three ingredients: 1) renaming states; 2) restricting to a stable subset of states; and 3) merging  compatible states. These ingredients are synthetized into two definitions (state renaming is implicit in both definitions).

\begin{definition}
Let $\CAA=(Q,R,V,V',f)$ be a stochastic CA. 
\begin{itemize}
\item if $i:Q'\rightarrow Q$ is an injective function such that $Y=(i(Q'))^\ZZ$ is $F$-stable (\textit{i.e.} $F(Y,R^\ZZ)\subseteq Y$) then the \emph{$i$-restriction} of $\CAA$ is the stochastic CA:
  \[\rest{i}{\CAA}=(Q',R,V,V',\rest{i}{f})\]
  where $\rest{i}{f}$ is the local function associated with the explicit global function $\rest{i}{F}$ such that, $\forall c\in I$, $\forall s\in R^\ZZ$, $\rest{i}{F}(c,s) = I^{-1}\circ F(I(c),s)$ where $I:Q'^\ZZ\rightarrow Q^\ZZ$ denotes the cell-by-cell extension of $i$;
\item if $\pi:Q\rightarrow Q'$ is surjective and $F$-compatible (\textit{s.t.} ${\Pi\circ F(c,s) = \Pi\circ F(c',s)}$ for all $s$ and all $c,c'$ such that $\Pi(c)=\Pi(c')$, where $\Pi:Q^\ZZ\rightarrow Q'^\ZZ$ is the cell-by-cell extension of $\pi$), then the \emph{$\pi$-projection} of $\CAA$ is the stochastic~CA:
  \[\proj{\pi}{\CAA}=(Q',R,V,V',\proj{\pi}{f})\]
  where $\proj{\pi}{f}$ is the local function associated with the explicit global function $\proj{\pi}{F}$ such that $\proj{\pi}{F}(c',s) = \Pi\circ F(c,s)$ where $c$ is any configuration in $\Pi^{-1}(c')$.
\end{itemize}
\end{definition}

If $i:Q'\rightarrow Q$ and $\pi:Q'\rightarrow Q''$ verify the required  stability and compatibility conditions, we denote by $\mix{\pi}{i}{\CAA}$ the $\pi$-projection of the $i$-restriction of \CAA{}.

\begin{definition}\label{def:locrel}
  Let $\CAA_1=(Q_1,R_1,V_1,V_1',f_1)$ and $\CAA_2=(Q_2,R_2,V_2,V_2',f_2)$ be two arbitrary stochastic CA. We define the following relations:
  \begin{itemize}
  \item $\CAA_1\SSUB\CAA_2$, $\CAA_1$ is a \emph{stochastic subautomaton} of $\CAA_2$, if there is some \mbox{$i$-restriction} of $\CAA_2$ such that $\STOC{F_1}=\STOC{\rest{i}{F_2}}$;
  \item $\CAA_1\SPROJ\CAA_2$, $\CAA_1$ is a \emph{stochastic factor} of $\CAA_2$, if there is some $\pi$-projection of $\CAA_2$ such that $\STOC{F_1}=\STOC{\proj{\pi}{F_2}}$;
    \DELETE{ \item $\CAA_1\NSUB\CAA_2$, $\CAA_1$ is a \emph{non-deterministic subautomaton} of $\CAA_2$, if there is some $i$-restriction of $\CAA_2$ such that $\NDET{F_1}=\NDET{\rest{i}{F_2}}$;
    \item $\CAA_1\NPROJ\CAA_2$, $\CAA_1$ is a \emph{non-deterministic factor} of $\CAA_2$, if there is some $\pi$-projection of $\CAA_2$ such that $\NDET{F_1}=\NDET{\proj{\pi}{F_2}}$.
    \end{itemize}
    Moreover, if $\CAA_1$ and $\CAA_2$ are deterministic, we also define
    \begin{itemize}
    \item $\CAA_1\DSUB\CAA_2$, $\CAA_1$ is a \emph{deterministic subautomaton} of $\CAA_2$, if there is some \mbox{$i$-restriction} of $\CAA_2$ such that $\DET{F_1}=\DET{\rest{i}{F_2}}$;
    \item $\CAA_1\DPROJ\CAA_2$, $\CAA_1$ is a \emph{deterministic factor} of $\CAA_2$, if there is some $\pi$-projection of $\CAA_2$ such that $\DET{F_1}=\DET{\proj{\pi}{F_2}}$.}
  \end{itemize}
  Similarly, we define $\NSUB$ and $\NPROJ$ (for non-deterministic global maps) and $\DSUB$ and $\DPROJ$ (for deterministic global maps).
  We also define the three relations $\DMIX$, $\NMIX$ and $\SMIX$ using projections of restrictions. For instance: $\CAA_1\SMIX\CAA_2$ if there are $i$ and $\pi$ such that $\STOC{F_1}=\STOC{\mix{\pi}{i}{F_2}}$.
\end{definition}


\subsubsection{Rescaling transformations.}

The transformations defined so far only allow to derive a finite number of CA from a given CA (up to renaming of the states) and thus induce only a finite number of dynamics. In particular, the size of the set of states, and the size of the neighborhood, can only decrease. Following the approach taken for classical deterministic CA, we now consider \emph{rescaling transformations}, which allow to increase the set of states, the neighborhood, etc\DELETE{ while preserving the dynamics, in fact preserving the complete information in the space-time diagram}. Rescaling transformations consist in: composing with a fixed translation, packing cells into fixed-size blocks, and iterating the rule a fixed number of times. Notice that since stochastic CA are composable, they are stable under rescaling operations, whereas \PCA{} are not.

The \emph{translation} $\shift{k}$ (for $k\in\ZZ$) is the deterministic CA whose deterministic global function verifies: ${\forall c,\forall z, \DET{\shift{k}}(c)_z = c_{z+k}}$.

Given any finite set $S$ and any $m\geq 1$, we define the bijective \emph{packing map} ${\bloc{m}: S^\ZZ\rightarrow
  \bigl(S^m\bigr)^\ZZ}$ by ${\bloc{m}(c)_z = (c_{mz},c_{mz+1},\ldots,c_{mz+m-1})}$
for all $c$ and $z$.

\begin{definition}
  Let $\CAA=(Q,R,V,V',f)$ be any stochastic CA. Let $m,t\geq 1$ and $k\in\ZZ$. The \emph{rescaling} of \CAA{} with parameters $(m,t,k)$ is the stochastic CA $\grp{\CAA}{m,t,k} = \bigl(Q^m, (R^m)^t, V_+, V_+',\grp{f}{m,t,k}\bigr)$ whose explicit global function $\grp{F}{m,t,k}$ is defined by:
  \[\grp{F}{m,t,k}(c,s) = \bloc{m}\circ \shift{k}\circ F^t(\debloc{m}(c),\debloc{m}(s^1),\ldots,\debloc{m}(s^t))\]
  where $s^1,\ldots,s^t\in (R^m)^\ZZ$ are the $t$ components of $s$ (s.t. $s^i_j = (s_j)_i$), and $V_+, V_+'$ the modified neighbourhoods following $b_m$.
\end{definition}

\subsection{Simulation Pre-Orders}

We can now define the general simulation relations.

\begin{definition}
  For each local relation $\somerel$ among the nine relations of Definition~\ref{def:locrel}, we define the associated simulation relation $\simu$ by
  \[\CAA_1\simu\CAA_2 \Leftrightarrow \exists m_1,m_2,t_1,t_2,k_1,k_2, \grp{\CAA_1}{m_1,t_1,k_1}\,\somerel\,\grp{\CAA_2}{m_2,t_2,k_2}\]
We therefore define nine simulation relations $\ssimui$, $\ssimus$, $\ssimum$, $\nsimui$, $\nsimus$, $\nsimum$, $\dsimui$, $\dsimus$ and $\dsimum$, where the subscript denotes the kind of local relation used (\textbf{i}njection, \textbf{$\pmb\pi$\hspace{-.2mm}}rojection or \textbf{m}ixed) and the superscript denotes the kind of global functions which are compared (\textbf{S}tochastic, \textbf{N}on-deterministic or \textbf{D}eterministic).

\DELETE{  To simplify notations and help the reader, we denote by 'i' (like injection) the local relations using restriction, by '$\pi$'  those using projection, and by 'm' (like mixed) those using restriction of projection. Hence, we define the nine following simulation relations:
  \begin{center}
    \small
    \begin{tabular}{r|c|c|c}
      & injection & surjection & mixed\\
      \hline
      &&&\\
      deterministic & $\dsimui$ & $\dsimus$ & $\dsimum$\\
      &&&\\
      non-deterministic & $\nsimui$ & $\nsimus$ & $\nsimum$\\
      &&&\\
      stochastic & $\ssimui$ & $\ssimus$ & $\ssimum$\\
    \end{tabular}
  \end{center}}
\end{definition}

\begin{lemma}\label{lem:restproj}
A restriction (resp. projection) of a restriction (resp. projection) of some stochastic CA \CAA{} is a restriction (resp. projection) of \CAA{}. 
Moreover, any restriction of a projection of \CAA{} is the projection of some restriction of \CAA{}.
\end{lemma}
\begin{proof}
  This is a straightforward generalization of the corresponding result in the classical deterministic settings. A detailed proof for the deterministic case appears in Theorem~2.1 of \cite{bulking2}. All arguments given in the proof are easily adaptable to our setting.
\end{proof}

The lemma above implies that any sequence of admissible restrictions and projections can be expressed as the projection of some restriction. 

From Lemma~\ref{lem:restproj} it follows that all local relations defined are transitive and reflexive. Moreover, the deterministic relations $\DSUB$ and $\DPROJ$ are exactly the same as those defined in the classical setting of deterministic CA \cite{bulking2}.

\begin{fact}\label{fac:simupreorders}
  All simulation relations $\ssimui$, $\ssimus$, $\ssimum$, $\nsimui$, $\nsimus$, $\nsimum$, $\dsimui$, $\dsimus$, $\dsimum$  are pre-orders.
\end{fact}

\begin{proof}
  It is sufficient to verify that for any local comparison relation $\somerel$:
  \begin{enumerate}
  \item $\somerel$ is compatible with rescalings, \textit{i.e.}
    \[\CAA_1\somerel\CAA_2\Rightarrow \grp{\CAA_1}{m,t,k}\somerel\grp{\CAA_2}{m,t,k}\]
  \item rescalings are commutative with respect to $\somerel$, \textit{i.e.}
    \[\grp{\grp{\CAA_1}{m,t,k}}{m',t',k'} \somerel \grp{\grp{\CAA_1}{m',t',k'}}{m,t,k}\]
  \end{enumerate}
  Both properties are straightforward from the definitions. Then, the transitivity of any simulation relation follows from the transitivity of the corresponding local comparison relation $\somerel$.
\end{proof}

Each stochastic pre-order is a refinement of the corresponding non-deterministic pre-order as shown by the following fact (straightforward corollary of Fact~\ref{fac:stocndet}).

\begin{fact}\label{fac:ifssimuthennsimu}
  If $\CAA_1\ssimui\CAA_2$ then $\CAA_1\nsimui\CAA_2$. The same is true for pre-orders $\ssimus$, $\ssimum$ and the corresponding (non-)deterministic pre-orders.
\end{fact}

Note that for any simulation relation $\simu$, ${\CAA_1\simu\CAA_2}$ means that two global functions are equal where one is obtained by applying \emph{only} space-time-diagram-preserving rescaling transformations to $\CAA_1$ (the simulated CA) and the other is obtained by applying both rescaling transformations and trimming operations to $\CAA_2$ (the simulator).

\subsection{Classifications of Stochastic Cellular Automata}

Simulation pre-orders can be seen as a tool to classify the behaviors of CA  \cite{bulking1,bulking2}. They can be used to formalize in a more precise way the empirical classes defined historically through experimentations. We now give some results on the structure induced on stochastic CA by this classification.

\paragraph{Ideals.} Some classes of stochastic CA may only simulate CA of their own class. This is the case of the deterministic CA and also of the class of the \emph{noisy} CA which are the CA $F$ such that ${\NDET{F}(c)=Q^\ZZ}$ for all~$c$. 

\begin{fact}\label{fact:ideals}
  Let $\simu$ be any non-deterministic or stochastic pre-order. Let $\CAA_1$ and $\CAA_2$ be stochastic CA such that $\CAA_1\simu\CAA_2$. If $\CAA_2$ is deterministic (resp. noisy) then $\CAA_1$ is deterministic (resp. noisy).
\end{fact}
\begin{proof}
  By Fact~\ref{fac:ifssimuthennsimu} it is sufficient to prove this for non-deterministic simulations. The property that the explicit global function is deterministic or noisy (\textit{i.e.} surjective on each configuration) is preserved by rescaling transformation. Hence it is sufficient to check that being deterministic or noisy is preserved by restriction and projection. This is straightforward for projection (because a projection is an onto map). Determinism is clearly preserved by restriction. Moreover, a noisy stochastic CA does not admit any non-trivial restriction because no subset of states is stable under iteration. Hence, the restriction of a noisy CA is necessarily itself (up to renaming of states) or the trivial CA with only one state. Both are noisy and the fact follows.
\end{proof}

\paragraph{Simulation of stochastic CA  by a \PCA{}.} Even if some stochastic CA cannot be expressed as a \PCA{} (because of potential local random correlation), each can be simulated by a particular \PCA{}. Each step is simulated by two steps: 1) each cell first copies its random symbol in its state so that 2) its neighbors read in its state its random symbol to complete the transition. \label{sec:ppcasimul}

\begin{theorem}\label{thm:pcasimu}
  For any stochastic CA $\AUTO A=(Q,R,V,V',f_A)$ there is a \PCA{} $\AUTO B$ such that $\CAA\ssimui\AUTO B$.
\end{theorem}

\begin{proof}
The idea is to simulate one step of \CAA{} by two steps of $\AUTO B$:
\begin{enumerate}
\item generate a random symbol locally and copy it to a component of states;
\item simulate a stochastic transition of \CAA{} reading states only and ignoring random symbols.
\end{enumerate}
Formally, let $\AUTO B=(Q_B,R,V,V',f_B)$ where $Q_B=Q\cup Q\times R$ and $f_B$ is any local function such that the associated explicit global function $F_B$ verifies:
\begin{enumerate}
\item for any ${c\in Q^\ZZ\subseteq Q_B^\ZZ}$ and any $s\in R^\ZZ$, ${\bigl(F_B(c,s)\bigr)_z = (c_z,s_z)}$
\item for any ${c\in (Q\times R)^\ZZ\subseteq Q_B^\ZZ}$ and any $s\in R^\ZZ$, ${F_B(c,s) = F_A(\pi_Q(c),\pi_R(c))}$ where $\pi_Q$ and $\pi_R$ are cell-by-cell projections on $Q$ and $R$ respectively.
\end{enumerate}
It is straightforward to check that ${\CAA\SSUB{\AUTO B}^2}$ with the restriction induced by the identity injection ${i : Q^\ZZ\rightarrow Q^\ZZ\subseteq Q_B^\ZZ}$.  
\end{proof}

Note that the restriction is essential in the above construction since the behavior is not specified (and no correct behavior can be specified) on configurations where states of type $Q$ and states of type $Q\times R$ are mixed. In particular it is \textbf{false} that the stochastic CA is the square of the \PCA{}; it is a restriction of that.

Still, one could think that we might achieve a simpler simulation by taking $Q_B=Q\times R$ and doing the two steps simultaneously so that ${F_B(c,s)}$ would be the cell by cell product of ${F_A(\pi_Q(c),\pi_R(c))}$ and $s$. But this does not work: for such a $\AUTO B$ there is generally  no restriction  nor projection nor combination of both able to reproduce the stochastic global function of $\AUTO A$. Indeed, if some $c$ and $s_1,s_2$ are such that ${F_A(c,s_1)\neq F_A(c,s_2)}$ there is no valid way to define  a corresponding configuration for $c$ in $F_B$ because the $Q$-component of states in $F_B$ depends only on the previous deterministic configuration, not on the random configuration. Then, one might see this impossibility as an argument against our formalism of simulation. Of course, many extensions of our definitions might be considered to allow more simulations between stochastic CA. However, we think that the random component of the simulated CA should never be used to determine which deterministic configuration of the simulator CA corresponds to which deterministic configuration of the simulated CA. Doing so would be like predicting the noise of a system to prepare the state of another system. In particular, we do not see any reasonable formal setting where $F_B$ defined as above would be able to simulate $F_A$. $F_A$ and $F_B$ might look like two syntactical variants of essentially the same object, but, as stochastic dynamical systems, they are very different. For instance, not every configuration can be reached from any configuration in $F_B$ whereas $F_A$  could have this property (\textit{i.e.} be a noisy stochastic CA).

We believe that a better understanding of the relationship between stochastic CA and \PCA{} should go through the following questions: is there a \PCA{} in any equivalence class induced by the pre-order $\ssimui$? is any stochastic CA $\ssimus$-simulated by some \PCA{}?

\section{Universality}
\label{sec:univ}

The quest for universal CA is as old as the model itself. Intrinsic universality has also a long story as reported in \cite{OllingerUnivhistory}. Our formalism of simulation allows to open the quest to stochastic cellular automata. 

Indeed, one of the main by-product of each simulation pre-order defined above is a notion of intrinsic universality. Formally, given some simulation pre-order $\simu$, a stochastic CA $\CAA$ is \emph{$\simu$-universal} if for any stochastic CA $\AUTO B$ we have ${\AUTO B\simu \CAA}$. When considering deterministic pre-orders, we recover the notions of universality already studied in literature for classical deterministic CA \cite{Ollinger6states,OllingerRichard4states,bulking2}.

\subsection{Negative results}

When considering non-deterministic or stochastic global functions, the random symbols are hidden. Still, the choice of the set of random symbols plays an important role in the global functions we can possibly obtain. We denote by $\PF(n)$ the set of the prime factors of~$n$. By extension, for a stochastic CA \CAA{} with set of random symbols $R$, we denote by $\PF(\CAA)$ the set $\PF(|R|)$. We have the following result:

\begin{lemma}\label{lem:primefactors}
  Let $\AUTO A_1=(Q,R_1,V_1,V_1',f_1)$ and $\AUTO A_2=(Q,R_2,V_2,V_2',f_2)$ be two stochastic CA with same set of states. If they are not deterministic and ${\STOC{F_1}=\STOC{F_2}}$ then $\PF({\AUTO A}_1)\cap\PF({\AUTO A}_2)\neq\emptyset$.
\end{lemma}
\begin{proof} If $\CAA_1$ is not deterministic, then there must exist some configuration $c\in Q^\ZZ$ and two configurations $y\neq y'$ such that $\{y,y'\}\subseteq\NDET{F_1}(c)$. So there are two disjoint cylinders $\cyl{u}{z}\cap\cyl{u'}{z} = \emptyset$ with $y\in\cyl{u}{z}$ and $y'\in\cyl{u'}{z}$. Therefore ${0<\bigl(\STOC{F}(c)\bigr)(\cyl{u}{z})<1}$. Besides, by definition of $\STOC{F}$, we have 
\[\bigl(\STOC{F_1}(c)\bigr)(\cyl{u}{z}) = \nu_1(\anevent^1_{c,\cyl{u}{z}}) = \frac{p}{q}<1\]
for some relatively prime numbers $p$ and $q$ (recall that $\nu_1$ is the uniform measure over $R_1^\ZZ$ and that ${\anevent^1_{c,\cyl{u}{z}}=\{s\in R_1^\ZZ : F_1(c,s)\in\cyl{u}{z}\}}$). Moreover ${\PF(q)\subseteq\PF(|R_1|)=\PF(\CAA_1)}$ since $\anevent^1_{c,\cyl{u}{z}}$ is a finite union of cylinders and since the $\nu_1$-measure of any cylinder is a rational of the form $\frac{a}{|R_1|^b}$ for some integers $a,b\geq 1$.
Now, by hypothesis, we have also
\[\bigl(\STOC{F_2}(c)\bigr)(\cyl{u}{z}) = \frac{p}{q}\]
and by a similar argument as above we deduce that $\PF(q)\subseteq\PF(\CAA_2)$. The lemma follows since $\PF(q)\neq\emptyset$ (because $\frac pq<1$).
\end{proof}

From Lemma \ref{lem:primefactors} it follows, surprisingly perhaps, that the random symbols of a stochastic CA limit its simulation power to stochastic CA that have compatible random symbols.
\begin{theorem}\label{thm:primefactors}
  Let $\simu$ be any stochastic simulation pre-order, and $\CAA_1$ and $\CAA_2$ two stochastic CA which are not deterministic. If $\CAA_1\simu\CAA_2$ then ${\PF({\AUTO A}_1)\cap\PF({\AUTO A}_2)\neq\emptyset}$.
\end{theorem}
\begin{proof} Trimming operations (restrictions and projections) do not modify the set of random symbols. Rescaling transformations modify the set of random symbols in the following way: $R\mapsto R^n$ for some integer $n$. Therefore such transformations preserve the set of prime factors $\PF(\CAA)$ of the considered CA \CAA. Moreover, rescaling transformations do not affect determinism: the rescaled version of a CA which is not deterministic cannot be deterministic. Hence, the relation $\CAA_1\simu\CAA_2$ implies an equality of stochastic global functions of two CA which have the same prime factors as $\CAA_1$ and $\CAA_2$ and one of which is not deterministic. Therefore none of them is deterministic and the theorem follows from lemma~\ref{lem:primefactors}.
\end{proof}
\DELETE{\begin{proof}(Sketch) First $\PF{\cdot}$ is invariant under any rescaling transformation. Then, if the distribution matches, there must be coupling of the random strings of some  finite lenght by Theorem~\ref{thm:coupling}. The existence of such a coupling implies that $|R_1|$ and $|R_2|$ must have a common prime factor. 
\end{proof}}
\DELETE{Given any stochastic pre-order $\simu$, a $\simu$-universal CA cannot be deterministic because determinism is preserved by simulation (Fact~\ref{fact:ideals}). Moreover, for any prime~$p$ there is some $\CAA_p$ with $\PF(\CAA_p)=\{p\}$. Hence Theorem~\ref{thm:primefactors},}

The consequence in terms of universality is immediate and breaks our hopes for a stochastic universality construction. 
\begin{corollary} \label{cor:no:stoc:universal}
  Let $\simu$ be any stochastic simulation pre-order. There is no \mbox{$\simu$-universal} stochastic CA.
\end{corollary}

\subsection{Positive results}


Still, the negative result of Corollary \ref{cor:no:stoc:universal} leaves open the possibility of partial universality constructions. We will now describe how to construct a stochastic CA which is $\nsimui$-universal (hence also $\nsimum$-universal; however note that the existence of a  $\nsimus$- or even of a $\dsimus$-universal is still open), and then draw the consequences.\\
Since we are not concerned with size optimization, we will use simple construction techniques using parallel Turing heads and table lookup as described for classical deterministic CA in \cite{OllingerUnivhistory}. More precisely, we construct a stochastic CA $\AUTO U=(Q_U,R_u,V_U,V_U',f_U)$ able to $\nsimui$-simulate any stochastic CA $\AUTO A=(Q,R,V,V',f)$ with no rescaling transformation on $\AUTO A$ and no shift in the rescaling of $\AUTO U$. Therefore each cell of $\AUTO A$ will be simulated by a block of $m$ cells of $\AUTO U$ and each step of $\CAA$ will be simulated by $t$ steps of $\AUTO U$ ($t$ and $m$ depend on \CAA{} and are to be determined later).\\
The blocks of $m$ cells have the following structure (the restriction in the pre-order handles the trimming of any invalid block):
\begin{center}\small\sf
    \begin{tabular}{|c|c|c|c|c|c|}
      \hline
      SYNC & transition table & $Q$-state & $R$-symbol & $Q$-states of neighbors & $R$-symbols of neighbors \\
      \hline
    \end{tabular}
\end{center}
where each part uses a fixed alphabet (independent of $Q$ and $R$) and only the width of each part may depend on \CAA{}\DELETE{ (note that the transition table encodes in particular the size of $Q,R,V,V'$)}. To each such block is attached a Turing head which will repeat cyclically a sequence of $4$ steps (sub-routines) described below. On a complete configuration made of such blocks there will be infinitely many such heads (one per block) executing these steps in parallel. Execution is synchronized at the end of each step (\textsf{SYNC} part) and such that two Turing heads never collide. Precisely, for some steps (2 and 4) the moves of all heads are rigorously identical (hence synchronous and without head collision). For some other steps (1 and 3), the sequence of moves of each head depend on the content of its corresponding block but these steps are always such that the head don't go outside the block (hence no risk of head collision) and they are synchronized at the end by the \textsf{SYNC} part which implements a small time countdown initialized to the maximum time needed to complete the step in the worst case. The parts holding $R$-symbols are initially empty (uniformly equal to some symbol) for each block. The $4$ steps are as follows:
\begin{enumerate}
\item generate \DELETE{surjectively }a string representing a random $R$-symbol in the \textsf{$R$-symbol} part using (possibly several) random $R_U$-symbols present in that part of the block;
\item copy the \textsf{$R$-symbol} part to the appropriate position in the \textsf{$R$-symbols of neighbors} part of each neighboring block. Do the same for \textsf{$Q$-state};
\item using information about $Q$-states and $R$-symbols in the block, find the corresponding entry in the transition table and update the \textsf{$Q$-states} part of the block accordingly;
\item clean \textsf{$R$-symbol} and \textsf{$R$-symbols of neighbors} parts (\textit{i.e.} write some uniform symbol everywhere).
\end{enumerate}

This construction scheme is very similar to the one used for classical deterministic CA but two points are important in our context:
\begin{itemize}
\item step $4$ is here to ensure that each configuration of \CAA{} has a canonical corresponding configuration of $\AUTO U$ made of blocks where the parts holding $R$-symbols is clean; (step~4 is required for the existence of the injection~$i$)
\item depending on the way we generate strings representing a $R$-symbols from strings of $R_U$-symbols in step $1$, we will obtain or not a uniform distribution over $R$ (recall Theorem~\ref{thm:primefactors}).
\end{itemize}

In the general case, we can always fix (by the means of the injection~$i$) a  width large enough for  parts containing the \textsf{$R$-symbols}  so that all $R$-symbols can be obtained (but with possibly different probabilities). We therefore obtain a universality result for non-deterministic simulations.

\begin{theorem}\label{thm:nondet:univ}
  Let $\simu$ be either $\nsimui$ or $\nsimum$. There exists a $\simu$-universal CA.
\end{theorem}

Note that this  $\simu$-universal CA is a \PCA, and we obtain thus a stronger version of the simulation mentioned in Section~\ref{sec:ppcasimul} page~\pageref{sec:ppcasimul}.

Now, if we are in a case where ${\PF(\CAA)\subseteq\PF(U)}$ then it is possible to choose a generation process in step $1$ such that each $R$-symbol is generated with the same probability. We therefore obtain an optimal partial universality construction for stochastic simulations.

\begin{theorem}\label{thm:stoc:univ}
  Let $\simu$ be either $\ssimui$ or $\ssimum$. For any finite set $P$ of prime numbers, there is a stochastic CA $\AUTO U_P$ such that for any stochastic CA \CAA: ${\PF(\CAA)\subseteq P\Rightarrow \CAA\simu\AUTO U_P.}$
Moreover  $\AUTO U_P$ is a \PCA.
\end{theorem}

\section{Open Problems}
\label{sec:open}

Intrinsic simulations has been proven to be a powerful tool to hierarchize behaviors in the deterministic world. In particular, the notion universal CA allows to formalize the concept of ``most complex'' CA as the ones concentrating ``all the possible behaviors'' within a given class \cite{bulking1,bulking2}.
The formalism and the notion of intrinsic simulation developed here for stochastic CA, enables us to export this classification tool to the stochastic world. In particular, it would be interesting to see whether our partial universality construction relates to experimentally observed classes, as in \cite{RST}. At the more theoretical level and amongst all the concrete questions raised by this article, the following ones are of particular interest:
\begin{itemize}
\item 
Is there for any stochastic \CAA{}, a \PCA{} $\AUTO B$ which is  ${\ssimui}$-equivalent to ${\AUTO A}$?
\item 
Is there for any stochastic \CAA{}, a \PCA{} $\AUTO B$ such that ${\CAA\ssimus\AUTO B}$?
\item 
Are there $\nsimus$-universal cellular automata?
\item  
Are universal CA the same for pre-order $\nsimui$ and $\nsimus$?
\end{itemize}

Our setting can also be generalized by taking any Bernouilli measure on the $R$-component (instead of the uniform measure). We believe that positive and negative results about universality essentially still hold but under a different form.

As noticed by an anonymous referee, there is an easy algorithm to decide whether to 1D CA have the same non-deterministic global function. We are currently working on an adaptation to decide equality of global stochastic functions.

\bibliographystyle{eptcs}
\bibliography{uspca}

\begin{thebibliography}{10}
\providecommand{\bibitemdeclare}[2]{}
\providecommand{\surnamestart}{}
\providecommand{\surnameend}{}
\providecommand{\urlprefix}{Available at }
\providecommand{\url}[1]{\texttt{#1}}
\providecommand{\href}[2]{\texttt{#2}}
\providecommand{\urlalt}[2]{\href{#1}{#2}}
\providecommand{\doi}[1]{doi:\urlalt{http://dx.doi.org/#1}{#1}}
\providecommand{\bibinfo}[2]{#2}

\bibitemdeclare{article}{ArrighiNUQCA}
\bibitem{ArrighiNUQCA}
\bibinfo{author}{P.~\surnamestart Arrighi\surnameend} \&
  \bibinfo{author}{J.~\surnamestart Grattage\surnameend}
  (\bibinfo{year}{2009}): \emph{\bibinfo{title}{{Intrinsically universal
  $n$-dimensional quantum cellular automata}}}.
\newblock {\sl \bibinfo{journal}{J. of Computer and Systems Sciences, to
  appear.}}

\bibitemdeclare{inproceedings}{DBLP:conf/cie/ArrighiFNT11}
\bibitem{DBLP:conf/cie/ArrighiFNT11}
\bibinfo{author}{Pablo \surnamestart Arrighi\surnameend},
  \bibinfo{author}{Renan \surnamestart Fargetton\surnameend},
  \bibinfo{author}{Vincent \surnamestart Nesme\surnameend} \&
  \bibinfo{author}{Eric \surnamestart Thierry\surnameend}
  (\bibinfo{year}{2011}): \emph{\bibinfo{title}{Applying Causality Principles
  to the Axiomatization of Probabilistic Cellular Automata}}.
\newblock In: {\sl \bibinfo{booktitle}{CiE}}, pp. \bibinfo{pages}{1--10}.
\newblock \urlprefix\url{http://dx.doi.org/10.1007/978-3-642-21875-0_1}.

\bibitemdeclare{unpublished}{Mairesse}
\bibitem{Mairesse}
\bibinfo{author}{Ana \surnamestart Busic\surnameend}, \bibinfo{author}{Jean
  \surnamestart Mairesse\surnameend} \& \bibinfo{author}{Irene \surnamestart
  Marcovici\surnameend} (\bibinfo{year}{2010}):
  \emph{\bibinfo{title}{Probabilistic cellular automata, invariant measures,
  and perfect sampling}}.
\newblock \bibinfo{note}{Pre-print arXiv:1010.3133}.

\bibitemdeclare{article}{bulking1}
\bibitem{bulking1}
\bibinfo{author}{Marianne \surnamestart Delorme\surnameend},
  \bibinfo{author}{Jacques \surnamestart Mazoyer\surnameend},
  \bibinfo{author}{Nicolas \surnamestart Ollinger\surnameend} \&
  \bibinfo{author}{Guillaume \surnamestart Theyssier\surnameend}
  (\bibinfo{year}{2011}): \emph{\bibinfo{title}{Bulking I: An abstract theory
  of bulking}}.
\newblock {\sl \bibinfo{journal}{Theor. Comput. Sci.}}
  \bibinfo{volume}{412}(\bibinfo{number}{30}), pp. \bibinfo{pages}{3866--3880}.
\newblock \urlprefix\url{http://dx.doi.org/10.1016/j.tcs.2011.02.023}.

\bibitemdeclare{article}{bulking2}
\bibitem{bulking2}
\bibinfo{author}{Marianne \surnamestart Delorme\surnameend},
  \bibinfo{author}{Jacques \surnamestart Mazoyer\surnameend},
  \bibinfo{author}{Nicolas \surnamestart Ollinger\surnameend} \&
  \bibinfo{author}{Guillaume \surnamestart Theyssier\surnameend}
  (\bibinfo{year}{2011}): \emph{\bibinfo{title}{Bulking II: Classifications of
  cellular automata}}.
\newblock {\sl \bibinfo{journal}{Theor. Comput. Sci.}}
  \bibinfo{volume}{412}(\bibinfo{number}{30}), pp. \bibinfo{pages}{3881--3905}.
\newblock \urlprefix\url{http://dx.doi.org/10.1016/j.tcs.2011.02.024}.

\bibitemdeclare{inproceedings}{Fates}
\bibitem{Fates}
\bibinfo{author}{Nazim \surnamestart Fates\surnameend} (\bibinfo{year}{2011}):
  \emph{\bibinfo{title}{{Stochastic Cellular Automata Solve the Density
  Classification Problem with an Arbitrary Precision}}}.
\newblock In \bibinfo{editor}{Thomas \surnamestart Schwentick\surnameend} \&
  \bibinfo{editor}{Christoph \surnamestart D{\"u}rr\surnameend}, editors: {\sl
  \bibinfo{booktitle}{STACS}}, {\sl
  \bibinfo{series}{LIPIcs}}~\bibinfo{volume}{9}, \bibinfo{publisher}{Schloss
  Dagstuhl - Leibniz-Zentrum fuer Informatik}, pp. \bibinfo{pages}{284--295},
  \doi{10.4230/LIPIcs.STACS.2011.284}.

\bibitemdeclare{inproceedings}{FatesRST06}
\bibitem{FatesRST06}
\bibinfo{author}{Nazim \surnamestart Fat{\`e}s\surnameend},
  \bibinfo{author}{Damien \surnamestart Regnault\surnameend},
  \bibinfo{author}{Nicolas \surnamestart Schabanel\surnameend} \&
  \bibinfo{author}{Eric \surnamestart Thierry\surnameend}
  (\bibinfo{year}{2006}): \emph{\bibinfo{title}{Asynchronous Behavior of
  Double-Quiescent Elementary Cellular Automata}}.
\newblock In: {\sl \bibinfo{booktitle}{LATIN}}, pp. \bibinfo{pages}{455--466},
  \doi{10.1007/11682462\_43}.

\bibitemdeclare{article}{Gacs}
\bibitem{Gacs}
\bibinfo{author}{P.~\surnamestart Gacs\surnameend} (\bibinfo{year}{2001}):
  \emph{\bibinfo{title}{{Reliable cellular automata with self-organization}}}.
\newblock {\sl \bibinfo{journal}{Journal of Statistical Physics}}
  \bibinfo{volume}{103}(\bibinfo{number}{1}), pp. \bibinfo{pages}{45--267}.

\bibitemdeclare{article}{Hedlund}
\bibitem{Hedlund}
\bibinfo{author}{G.~A. \surnamestart Hedlund\surnameend}
  (\bibinfo{year}{1969}): \emph{\bibinfo{title}{Endormorphisms and
  automorphisms of the shift dynamical system}}.
\newblock {\sl \bibinfo{journal}{Mathematical Systems Theory}}
  \bibinfo{volume}{3}, pp. \bibinfo{pages}{320--375}.

\bibitemdeclare{article}{Kari94}
\bibitem{Kari94}
\bibinfo{author}{Jarkko \surnamestart Kari\surnameend} (\bibinfo{year}{1994}):
  \emph{\bibinfo{title}{Reversibility and Surjectivity Problems of Cellular
  Automata}}.
\newblock {\sl \bibinfo{journal}{J. Comput. Syst. Sci.}}
  \bibinfo{volume}{48}(\bibinfo{number}{1}), pp. \bibinfo{pages}{149--182}.
\newblock \urlprefix\url{http://dx.doi.org/10.1016/S0022-0000(05)80025-X}.

\bibitemdeclare{book}{kurkabook}
\bibitem{kurkabook}
\bibinfo{author}{P.~\surnamestart K\r{u}rka\surnameend} (\bibinfo{year}{2003}):
  \emph{\bibinfo{title}{Topological and symbolic dynamics}}.
\newblock \bibinfo{publisher}{Soci\'et\'e Math\'ematique de France}.

\bibitemdeclare{article}{maki}
\bibitem{maki}
\bibinfo{author}{A.~\surnamestart Maruoka\surnameend} \&
  \bibinfo{author}{M.~\surnamestart Kimura\surnameend} (\bibinfo{year}{1976}):
  \emph{\bibinfo{title}{Condition for {I}njectivity of {G}lobal {M}aps for
  {T}essellation {A}utomata}}.
\newblock {\sl \bibinfo{journal}{Information and Control}}
  \bibinfo{volume}{32}, pp. \bibinfo{pages}{158--162}.

\bibitemdeclare{phdthesis}{phdollinger}
\bibitem{phdollinger}
\bibinfo{author}{N.~\surnamestart Ollinger\surnameend} (\bibinfo{year}{2002}):
  \emph{\bibinfo{title}{Automates cellulaires : structures}}.
\newblock Ph.D. thesis, \bibinfo{school}{\'Ecole Normale Sup\'erieure de Lyon}.

\bibitemdeclare{inproceedings}{Ollinger6states}
\bibitem{Ollinger6states}
\bibinfo{author}{Nicolas \surnamestart Ollinger\surnameend}
  (\bibinfo{year}{2002}): \emph{\bibinfo{title}{The Quest for Small Universal
  Cellular Automata}}.
\newblock In: {\sl \bibinfo{booktitle}{ICALP}}, pp. \bibinfo{pages}{318--329}.
\newblock \urlprefix\url{http://dx.doi.org/10.1007/3-540-45465-9_28}.

\bibitemdeclare{inproceedings}{OllingerUnivhistory}
\bibitem{OllingerUnivhistory}
\bibinfo{author}{Nicolas \surnamestart Ollinger\surnameend}
  (\bibinfo{year}{2008}): \emph{\bibinfo{title}{Universalities in cellular
  automata a (short) survey}}.
\newblock In: {\sl \bibinfo{booktitle}{JAC}}, pp. \bibinfo{pages}{102--118}.

\bibitemdeclare{article}{OllingerRichard4states}
\bibitem{OllingerRichard4states}
\bibinfo{author}{Nicolas \surnamestart Ollinger\surnameend} \&
  \bibinfo{author}{Ga{\'e}tan \surnamestart Richard\surnameend}
  (\bibinfo{year}{2011}): \emph{\bibinfo{title}{Four states are enough!}}
\newblock {\sl \bibinfo{journal}{Theor. Comput. Sci.}}
  \bibinfo{volume}{412}(\bibinfo{number}{1-2}), pp. \bibinfo{pages}{22--32}.
\newblock \urlprefix\url{http://dx.doi.org/10.1016/j.tcs.2010.08.018}.

\bibitemdeclare{incollection}{Pivato09}
\bibitem{Pivato09}
\bibinfo{author}{Marcus \surnamestart Pivato\surnameend}
  (\bibinfo{year}{2009}): \emph{\bibinfo{title}{Ergodic Theory of Cellular
  Automata}}.
\newblock In: {\sl \bibinfo{booktitle}{Encyclopedia of Complexity and Systems
  Science}}, \bibinfo{publisher}{Springer}, pp. \bibinfo{pages}{2980--3015},
  \doi{10.1007/978-0-387-30440-3\_178}.

\bibitemdeclare{phdthesis}{phdrapaport}
\bibitem{phdrapaport}
\bibinfo{author}{I.~\surnamestart Rapaport\surnameend} (\bibinfo{year}{1998}):
  \emph{\bibinfo{title}{Inducing an order on cellular automata by a grouping
  operation}}.
\newblock Ph.D. thesis, \bibinfo{school}{\'Ecole Normale Sup\'erieure de Lyon}.

\bibitemdeclare{article}{RST}
\bibitem{RST}
\bibinfo{author}{D.~\surnamestart Regnault\surnameend},
  \bibinfo{author}{N.~\surnamestart Schabanel\surnameend} \&
  \bibinfo{author}{{\'E}.~\surnamestart Thierry\surnameend}
  (\bibinfo{year}{2009}): \emph{\bibinfo{title}{{Progresses in the analysis of
  stochastic 2D cellular automata: a study of asynchronous 2D minority}}}.
\newblock {\sl \bibinfo{journal}{Theoretical Computer Science}}
  \bibinfo{volume}{410}(\bibinfo{number}{47-49}), pp.
  \bibinfo{pages}{4844--4855}.

\bibitemdeclare{phdthesis}{phdtheyssier}
\bibitem{phdtheyssier}
\bibinfo{author}{G.~\surnamestart Theyssier\surnameend} (\bibinfo{year}{2005}):
  \emph{\bibinfo{title}{Automates cellulaires : un mod\`ele de complexit\'es}}.
\newblock Ph.D. thesis, \bibinfo{school}{\'Ecole Normale Sup\'erieure de Lyon}.

\bibitemdeclare{article}{Toom}
\bibitem{Toom}
\bibinfo{author}{A.~\surnamestart Toom\surnameend} (\bibinfo{year}{1995}):
  \emph{\bibinfo{title}{{Cellular automata with errors: Problems for students
  of probability}}}.
\newblock {\sl \bibinfo{journal}{Topics in Contemporary Probability and Its
  Applications}}, pp. \bibinfo{pages}{117--157}.

\end{thebibliography}




\end{document}